\documentclass[final,nomarks,pdftex]{dmtcs-episciences}

\usepackage[utf8]{inputenc}

\usepackage{amsthm,amsmath,enumerate}

\theoremstyle{plain}
\newtheorem{theorem}{Theorem}
\newtheorem{lemma}[theorem]{Lemma}
\newtheorem{corollary}[theorem]{Corollary}
\newtheorem{proposition}[theorem]{Proposition}
\newtheorem{observation}[theorem]{Observation}
\newtheorem{case}{Case}

\newtheorem{caseWithRef}{Case}

\newtheorem{construction}[theorem]{Construction}

\newcommand{\overbar}[1]{\mkern 5mu\overline{\mkern-5mu#1\mkern-5mu}\mkern 5mu}

\author[Tom\'a\v{s} Dvo\v{r}\'ak]{Tom\'a\v{s} Dvo\v{r}\'ak\hfil\hfil\hfil\hfill\hfil}

\title[Quadratic matchings extend to long cycles in hypercubes]
{Matchings of quadratic size extend to~long~cycles in hypercubes\,\begin{NoHyper}\thanks{%
Financial support from the Czech Science Foundation under the grant GA14-10799S is gracefully acknowledged.}\end{NoHyper}
}

\affiliation{Faculty of Mathematics and Physics, Charles University in Prague, Czech Republic\hfill\hfil\hfil\hfil}
\received{2015-11-23}
\accepted{2016-8-1}

\keywords{ Gray code, Hamiltonian cycle, hypercube, long cycle, matching, Ruskey and Savage problem}

\begin{document}

\publicationdetails{18}{2016}{3}{12}{1336}

\maketitle

\begin{abstract}

Ruskey and Savage in 1993 asked whether every matching in a hypercube can be extended to a Hamiltonian cycle. 
A~positive answer is known for perfect matchings, but the general case has been resolved only for matchings of linear size. In this paper we show that there is a quadratic function $q(n)$ such that every matching in the $n$-dimensional hypercube of size at most $q(n)$ may be extended to a cycle which covers at least $\frac34$ of the vertices.

\end{abstract}

\section{Introduction}
Frank Ruskey and Carla Savage in 1993 formulated the following problem \cite{RuSa}:  Does every matching in a hypercube extend to a Hamiltonian cycle? More than two decades have passed and the question still remains open. 
It may be of interest that a complementary problem on the existence of a Hamiltonian cycle avoiding a given matching in a hypercube has been already resolved \cite{DDGS09}.
  
An important step towards a solution to the Ruskey-Savage problem was made in 2007 by Fink who answered the question affirmatively for every perfect matching \cite{Fi07}.
Note that this  implies a~positive solution for every matching that extends to a perfect one, which includes e.g.~every induced matching \cite{VaWe}. However, this result does not immediately provide a~complete answer to the original question, as hypercubes contain matchings that are maximal with respect to inclusion, but still not perfect. Actually, to determine the minimum size of such a matching is another long-standing problem \cite{Fo}.  The simplicity and elegance of Fink’s method inspired several generalizations \cite{AAAHST15,FeSu09,Gre09}, but none of them addresses the extendability of imperfect matchings.

As far as arbitrary matchings are concerned, there are only partial results that deal with matchings of linear size.
The author of the current paper showed that a set $\mathcal{P}$ of at most $2n-3$ edges of the $n$-dimensional hypercube extends to a Hamiltonian cycle iff $\mathcal{P}$ induces a linear forest \cite{d-05}. This bound is sharp, as for every $n\ge3$ there is a non-extendable linear forest of $2n-2$ edges \cite{ck-05}. Of course, these edges do not form a matching, so this does not imply a negative answer to the Ruskey-Savage problem. In the case when $\mathcal{P}$ is a matching, the bound on $|\mathcal{P}|$ was improved to $3n-10$ by Wang and Zhang \cite{WaZha15}.

The purpose of this paper is to derive a quadratic upper bound on the size of a~matching that extends to a cycle which covers at least $\frac34$ of the vertices of the hypercube.
Our result is based on an inductive construction which combines a refinement of Fink's method for perfect matchings \cite{Fi09} with a lemma on hypercube partitioning due to Wiener \cite{Wie}.

\section{Preliminaries}
The graph-theoretic terms used in this paper but not defined below may be found e.g. in \cite{Bo}.
Throughout the paper, $n$ always denotes a positive integer while $[n]$ stands for the set $\{1,2,\dots,n\}$.

Vertex and edge sets of a graph G are denoted by $V(G)$ and $E(G)$, respectively. 
A~sequence $a=x_1,x_2,\dots,x_{n+1}=b$ of pairwise distinct vertices such
that $x_i$ and $x_{i+1}$ are adjacent for all $i\in[n]$ is a~\emph{path
between} $a$ and $b$ of \emph{length} $n$. 
We denote such a~path and its vertices by $P_{ab}$ and $V(P_{ab})$, respectively.
Let $P_{ab}$ and $P_{bc}$ be paths such that $V(P_{ab})\cap V(P_{bc})=\{b\}$. Then $P_{ab}+P_{bc}$ denotes the path between  $a$ and $c$, obtained as a~concatenation of $P_{ab}$
with $P_{bc}$ (where $b$ is taken only once). Observe that the operation $+$
is associative.
A \emph{cycle} of length $n$ is a sequence $x_1,x_2\dots,x_{n}$ of pairwise distinct vertices such that $x_1$ is adjacent to $x_{n}$ and $x_i$ is adjacent to $x_{i+1}$ for all $i\in[n]$. The sets of vertices $x_1,x_2\dots,x_{n}$ and edges $x_1x_2,x_2x_3,\dots,x_{n}x_1$ of a cycle $C$ are denoted by $V(C)$ and $E(C)$, respectively.

In this paper we deal with the $n$-dimensional hypercube $Q_n$ which is a graph with all $n$-bit strings as vertices, an edge joining two vertices whenever they differ in a single bit. 
Given a string $v=v_1v_2\dots v_n$ and a set $D\subseteq[n]$, we use $v_D$
to denote the string  $v_{j_1}v_{j_2}\dots v_{j_{d}}$ where $j_1,j_2,\dots,j_{d}$
is an increasing sequence of all elements of $D$.
Given a set $D\subseteq[n]$ of size $d$ and a vertex $u\in V(Q_{n-d})$, the
subgraph of $Q_n$ induced by the vertex set
$\{v\in V(Q_n)\mid v_{\bar D}=u\}$ where $\bar D=[n]\setminus D$ is
denoted by $Q_D(u)$ and called a \emph{subcube of dimension} $d$.
Subcubes $Q_D(u)$ and $Q_D(v)$ are called \emph{adjacent} if $u$
and $v$ are adjacent vertices of $Q_{n-d}$.
Note that
\begin{itemize}
\item $Q_D(u)$ is isomorphic to $Q_d$,
\item if $Q_D(u)$ and $Q_D(v)$ are adjacent subcubes, then an arbitrary
vertex in one of them has a unique neighbor in the other.
\end{itemize}
Given a set $\mathcal{S}\subseteq V(Q_n)$, $\mathcal{S}_D(u)$ denotes the set $\mathcal{S}\cap V(Q_D(u))$. 

To engineer our inductive construction, we employ the following result on hypercube partitioning due to Wiener \cite[Theorem~2.5]{Wie}, see also \cite[Section~4]{Wie13} for the proof. Although it was originally stated for set systems, here we provide an equivalent formulation using the terminology introduced above.
\begin{theorem}[\cite{Wie}]
\label{thm:wiener}
Let $\mathcal{S}$ be a set of vertices of $Q_n$ of size $s$
with $s\geq 2n$ and $d=\lceil\frac{n^2}{2s-n-2}\rceil$. Then there is a
set $D\subseteq[n]$ such that $|D|=d$ and $|\mathcal{S}_D(u)|\le d+1$
for every $u\in\{0,1\}^{n-d}$.
\end{theorem}
Let $K(Q_n)$ denote the complete graph on the set of vertices of $Q_n$. The set $\dim(uv)$ of \emph{dimensions}  of an edge $uv\in E(K(Q_n))$, $u=u_1\dots u_n, v =v_1\dots v_n$, is defined by $\dim(uv)=\{i\in[n]\mid u_i\ne v_i\}$. An edge $uv$ is called \emph{short} if $|\dim(uv)|=1$ and \emph{long} otherwise.  
For a vertex $v\in V(Q_n)$ let $v^d$ denote the vertex of $Q_n$ such that $vv^d$ is a short edge of dimension $d$. 

A \emph{matching} is a set of pairwise non-adjacent edges. Given a matching $M$, we use $\bigcup M$ to denote the set of all vertices incident with an edge of $M$.
We say that a matching $M$ in $K(Q_n)$ is \emph{$d$-saturated} if every short edge of dimension $d$ not in $M$ is adjacent to some edge of $M$.

Note that removing all edges of some fixed dimension $d$ splits $Q_n$ into two $(n-1)$-dimensional subcubes $Q^0=Q_{\overbar{\{d\}}}(0)$ and $Q^1=Q_{\overbar{\{d\}}}(1)$. We use  $M^i_d$ to denote $M\cap E(K(Q^i))$ for both $i\in\{0,1\}$ and $M^2_d=M\setminus(M^0_d\cup M^1_d)$. 
\begin{observation}
\label{obser:saturation}
Let $d\in[n]$ and $M$ be a matching in $K(Q_n)$ containing $s$ short edges of dimension $d$. If $2|M|-s<2^{n-1}$, then $M$ is not $d$-saturated.
\end{observation}
\begin{proof}
Put $Q^0=Q_{\overbar{\{d\}}}(0)$. 
Since $|\{u\in V(Q^0)\mid \{u,u^d\}\cap\bigcup M\ne\emptyset\}|=2|M|-s<2^{n-1}=|V(Q^0)|$, there must be a $u$ such that  $\{u,u^d\}\cap\bigcup M=\emptyset$. Hence $M$ is not $d$-saturated.
\end{proof}

The following properties of hypercubes shall be useful later.
\begin{proposition}
\label{prop:old}
Let $n\ge2$. Then
\begin{enumerate}[\upshape(1)]
  \item 
  \label{prop:old-separable-edges}
  {\rm\cite{DDGS09}} for every matching $M$ of size two in $Q_n$ there is $d\in[n]$ such that $|M^0_d|=|M^1_d|=1$, 
  \item 
  \label{prop:old-fink}
  {\rm\cite{Fi07}} every perfect matching of $K(Q_n)$ can be extended by short edges to a Hamiltonian cycle,
  \item
  \label{prop:old-wang}
  {\rm \cite{WaZha}}
every matching of size at most $2n-1$ in $Q_n$ can be extended to a Hamiltonian cycle.
\end{enumerate}
\end{proposition}
\section{Construction}
The following construction is a refinement of Fink's method \cite[Proof of Theorem~3]{Fi09} which was originally devised to extend perfect matchings in hypercubes to Hamiltonian cycles.
\begin{construction}
\label{construction}
Let $M$ be a matching in $K(Q_n)$, $n\ge2$, $d\in[n]$.\begin{enumerate}[\upshape1.]
\item 
Split $Q_n$ into subcubes $Q^0=Q_{\overbar{\{d\}}}(0)$ and $Q^1=Q_{\overbar{\{d\}}}(1)$.
\item
\label{construction:odd}
If $|M^2_d|$ is odd, select an edge of $E(Q_n)\setminus M$ of dimension $d$ which is not adjacent to any edge of $M$ and add it to $M^2_d$.
\item
\label{construction:empty}
If $M^2_d$ is empty, select two edges of $E(Q_n)\setminus M$ of dimension $d$ which are not adjacent to any edge of $M$ and add them to $M^2_d$.
\item
\label{construction:left}
Let $M^2_d=\{u_1v_1,u_2v_2,\dots,u_kv_k\}$ with $u_i\in V(Q^0)$ and $v_i\in V(Q^1)$ for all $i\in[k]$. \\
Select a perfect matching $P^0$ in the subgraph of $K(Q^0)$ induced by $u_1, u_2, \dots, u_k$. \\
Let $C^0$ be a cycle in $K(Q^0)$ such that $E(C^0)=M^0_d\cup P^0\cup S^0$ where $S^0\subseteq E(Q^0)$.
\item
\label{construction:right}
The removal of edges of $P^0$ breaks $C^0$ into pairwise disjoint paths, and we can without loss of generality assume that these are paths $P_{u_1u_2},P_{u_3u_4},\dots,P_{u_{k-1}u_k}$.  \\
Put $P^1=\{v_1v_2,v_3v_4,\dots,v_{k-1}v_k\}$. \\
Let $C^1$ be a cycle in $K(Q^1)$ such that $E(C^1)=M^1_d\cup P^1\cup S^1$ where $S^1\subseteq E(Q^1)$. 
\end{enumerate}
\end{construction}
Note that Construction~\ref{construction} is non-deterministic in the sense that edges selected in steps~\ref{construction:odd} and \ref{construction:empty} as well as a matching $P^0$ and cycles $C^0$ and $C^1$ formed in steps \ref{construction:left} and \ref{construction:right} may not be unique or may not even exist. If $C^0$ and $C^1$ are cycles of $K(Q^0)$ and $K(Q^1)$, respectively, defined in steps \ref{construction:left} and \ref{construction:right}, where $M^i_d$, $P^i$, $S^i$, $i\in\{0,1\}$ used in their definition are obtained by some execution of Construction~\ref{construction}, then the pair $(C^0,C^1)$ is called a \emph{$d$-extension} of the matching $M$ in $K(Q_n)$.
\begin{observation}
\label{observation:on-construction}
If $(C^0,C^1)$ is a $d$-extension of a matching $M$ in $K(Q_n)$, then there is a~cycle $C$ in~$K(Q_n)$ such that $E(C)=M\cup S$ where $S\subseteq E(Q_n)$ and $|V(C)|=|V(C^0)|+|V(C^1)|$.
\end{observation}
\begin{proof}
Referring to the notation of Construction~\ref{construction}, replace each edge $v_iv_j\in P^1$ in $C^1$ with the path $v_i+P_{u_iu_j}+v_j$. This extends $C^1$ to a cycle $C$ in $K(Q_n)$ such that $E(C)=M\cup S^0\cup S^1\cup S^2$ where $S^2$ consists of edges added in Step \ref{construction:odd} or \ref{construction:empty}. Then $S^0\cup S^1\cup S^2\subseteq E(Q_n)$ and $|V(C)|=|V(C^0)\,\mathbin{\dot{\cup}}\,V(C^1)|=|V(C^0)|+|V(C^1)|$. 
 \end{proof}
\section{Initial step}
The following application of Construction~\ref{construction} shall be useful as an initial step for the inductive proof of our main result.
\begin{lemma}
\label{lemma:n}
Every matching $M$ in $K(Q_n)$, $|M|\le n\ge2$, may be extended by short edges to a~cycle of~length at least $\frac{3}{4}|V(Q_n)|$. 
\end{lemma}
\begin{proof}
We argue by induction on $n$. The case $n=2$ may be verified by a direct inspection.
Next assume that $n>2$ and that $M\ne\emptyset$, for otherwise we may use an arbitrary Hamiltonian cycle of $Q_n$. 
\begin{case}
$n=3$ and for every $d\in[3]$, 
either $|M^2_d|=\emptyset$,
or $|M^2_d|$ is odd and $M$ is $d$-saturated.
\end{case}
Recall that if $M$ is $d$-saturated, Observation~\ref{obser:saturation} says that $|M|\ge2+s/2$ where $s$ is the number of short edges of dimension $d$. Consequently, $|M|=2$ would mean that $M$ consists of two long edges, which together with $n=3$ implies that $|M^2_d|=2$ for some $d$, contrary to our assumption. Hence it must be the case that $|M|=3$. If all these three edges are short, $M$ extends to a Hamiltonian cycle by part \eqref{prop:old-wang} of~Proposition~\ref{prop:old}. If at least one edge is long, then there is $d\in[n]$ such that $|M^2_d|>1$ and therefore, by our assumption, $|M^2_d|=3$.

Hence we can assume that $M=M^2_d=\{u_iv_i\mid i\in[3]\}$ where $u_i$ lies in $Q^0=Q_{\overbar{\{d\}}}(0)$ for all $i\in[3]$. Since $M$ is $d$-saturated, one of these edges, say $u_1v_1$, must have the property that $v_1^d$ is not incident with any edge of $M$. Then $P^0=\{u_1v_1^d,u_2u_3$\} forms a perfect matching of $K(Q^0)$ and therefore, by the induction hypothesis, it may be extended by short edges to a Hamiltonian cycle $C^0$ of $Q^0$. Since $Q^1=Q_{\overbar{\{d\}}}(1)$ is a cycle, there exists a path $P_{u_2^du_3^d}$ in $Q^1$ avoiding $v_1$.
Replacing edges $u_1v_1^d$ and $u_2u_3$ of~$C^0$ with the paths $u_1,v_1,v_1^d$ and $u_2+P_{u_2^du_3^d}+u_3$, respectively, we extend $C^0$ to a cycle $C$ of length at least $7>\frac34|V(Q_4)|$ such that $E(C)=M\cup S$ where $S\subseteq E(Q_3)$ as required.
\begin{case}
$n=4$ and for every $d\in[4]$, 
either $|M^2_d|=\emptyset$,
or $|M^2_d|$ is odd and $M$ is $d$-saturated.
\end{case}
Recall that $M\ne\emptyset$ and therefore there must be a $d'$ such that $M$ is $d'$-saturated. By Observation~\ref{obser:saturation}, $|M|\ge4+s/2$ where $s$ is the number of short edges of $M$ of dimension $d'$. Since we assume that $|M|\le4$, it follows that $|M|=4$ and every edge $e\in M$ with $d'\in\dim(e)$ is long. But then the are distinct edges $e,e'\in M$ such that $\dim(e)$ and $\dim(e')$ share the same dimension $d$, which means that $|M^2_{d}|=3$.

Hence we can without loss of generality assume that $M=\{u_iv_i\mid i\in[4]\}$ where $u_4v_4\in M^0_d$ while $u_iv_i\in M^2_d$ where $u_i$ lies in $Q^0=Q_{\overbar{\{d\}}}(0)$ for all $i\in[3]$. 
Since $M$ is $d$-saturated, $P^0=\{u_1v_1^d,u_2v_2^d,u_3v_3^d,u_4v_4$\} forms a perfect matching of $K(Q^0)$. By part \eqref{prop:old-fink} of Proposition~\ref{prop:old}, $P^0$ may be extended by short edges to a Hamiltonian cycle $C^0$ of $Q^0$.  Since $|E(C^0)\setminus E(P^0)|=4$, there is an edge $uv\in E(C^0)\setminus E(P^0)$ which is not incident with $v_i^d$ for any $i\in[3]$. Replacing edges $uv$ and $u_iv_i^d$ of $C^0$ with the paths $u,u^d,v^d,v$ and $u_i,v_i,v_i^d$ for all $i\in[3]$, respectively, we extend $C^0$ to a cycle $C$ of length $13>\frac34|V(Q_4)|$ such that $E(C)=M\cup S$ where $S\subseteq E(Q_4)$ as required.

\begin{case}
Either $n\ge5$, or $n\in\{3,4\}$ and there is $d\in[n]$ such that $M^2_d\ne\emptyset$ and if $M$ is $d$-saturated then $|M^2_d|$ is even. 
\end{case}
We show that in this case there is a $d\in[n]$ such that  a $d$-extension of $M$ exists.
If $3\le n\le 4$, let $d$ be such that either $|M^2_d|$ is even and positive, or $|M^2_d|$ is odd and $M$ is not $d$-saturated. If $n\ge5$, select $d\in[n]$ in the following way:
If $|M|<n$, let $d$ be an arbitrary dimension of some edge of $M$. If $|M|=n$ and for each $i\in[n]$ there is exactly one edge $e\in M$ with $i\in\dim(e)$, then there must be two distinct short edges $e_1,e_2\in M$. Then use part \eqref{prop:old-separable-edges} of Proposition~\ref{prop:old} to select $d$ in such a way that these edges belong to distinct subcubes  $Q^0=Q_{\overbar{\{d\}}}(0)$ and $Q^1=Q_{\overbar{\{d\}}}(1)$. If none of the previous cases applies, $d$ may be selected such that $M$ contains at least two edges of dimension $d$.

Now verify the validity of all steps of Construction~\ref{construction}, using  the notation introduced there. First inspect Step~\ref{construction:odd}, i.e. the case when $|M^2_d|$ is odd. If $3\le n\le 4$, then $M$ is not $d$-saturated by the above choice of $d$. For $n\ge5$ we have
$
2|M|\le 2n<2^{n-1},
$
which means that by Observation~\ref{obser:saturation}, $M$ is not $d$-saturated either. It follows that Step~\ref{construction:odd} may be safely performed.  

Since $M^2_d\ne\emptyset$, Step~\ref{construction:empty} is irrelevant, and hence it remains to verify the validity of Steps~\ref{construction:left} and \ref{construction:right}. Note that both $M^0_d\cup P^0$ and $M^1_d\cup P^1$ are matchings in $K(Q^0)$ and $K(Q^1)$, respectively, while our choice of $d$ guarantees that their size does not exceed $n-1$ unless $n=4$ when it may be equal to $n$. By the induction hypothesis (or by part \eqref{prop:old-fink} of Proposition~\ref{prop:old} in the case that $n=4$), these matchings may be extended by short edges to cycles $C^i$ of lengths at least $\frac{3}{4}|V(Q^i)|$ for both $i\in\{0,1\}$. The desired cycle of length at least $\frac{3}{4}|V(Q_n)|$ extending $M$ by short edges then exists by Observation~\ref{observation:on-construction}.
\end{proof}
\section{Results}
We start with another application of Construction~\ref{construction}. This time, the assumption on matching size is replaced with a requirement of even distribution of its edges among 4-dimensional subcubes. 
\begin{lemma}
  \label{lemma:induction}
Let $M$ be a matching in $K(Q_n)$, $n\ge4$,  and $D\subseteq[n]$ such that $|D|=4$ and $|\bigcup M_D(u)|\le7$ for every $u\in\{0,1\}^{n-4}$.  Moreover, there are at most two distinct $u,v\in\{0,1\}^{n-4}$ such that $|\bigcup M_D(u)|,|\bigcup M_D(v)|\in\{6,7\}$, and if they both exist, then the subcubes $Q_D(u)$ and $Q_D(v)$ are adjacent. 
Then there is a set $S\subseteq E(Q_n)$ such that $M\cup S$  forms a cycle in $K(Q_n)$ of length at least $\frac{3}{4}|V(Q_n)|$. 
\end{lemma}

\begin{proof}
We argue by induction on $n$. For $n=4$ we have $|\bigcup M|=|\bigcup M_D(u)|\le7$, which means that $|M|\le3$ and the statement follows from Lemma~\ref{lemma:n}. For $n>4$ we show that there is a $d$-extension of $M$ for certain $d$.  To that end, consider four cases. 
\begin{caseWithRef}
\label{lemma:induction-A}
$n>5$ and there are $u\ne v\in\{0,1\}^{n-4}$ such that $|\bigcup M_D(u)|,|\bigcup M_D(v)|\in\{6,7\}$.
 \end{caseWithRef}
\begin{caseWithRef}
\label{lemma:induction-B}
$n>5$ and there  is exactly one $u\in\{0,1\}^{n-4}$ such that $|\bigcup M_D(u)|\in\{6,7\}$.
 \end{caseWithRef}
\begin{caseWithRef}
\label{lemma:induction-C}
$n>5$ and $|\bigcup M_D(u)|\le5$ for every $u\in\{0,1\}^{n-4}$.
\end{caseWithRef}
\begin{caseWithRef}
\label{lemma:induction-D}
$n=5$.
\end{caseWithRef}
Note that in Case~\ref{lemma:induction-A}, the subcubes $Q_D(u)$ and $Q_D(v)$ are adjacent by our assumption, and hence we can select $d$ as the dimension of the edge $uv$. Otherwise $d$ may be an arbitrary element of $\bar{D}$ (in Case~\ref{lemma:induction-D} the only one).

Now go through the steps of Construction~\ref{construction}, using  the notation introduced there. 
To perform Step~\ref{construction:odd} or \ref{construction:empty}, we need to find one or two non-adjacent short edges of dimension $d$, not intersecting any edge of $M$. To that end, select $w\in\{0,1\}^{n-4}$ such that
\begin{itemize}
 \item
 in Case~\ref{lemma:induction-A}, subcube $Q_D(w)$ is adjacent to $Q_D(u)$ but different from $Q_D(v)$,
 \item
 in Case~\ref{lemma:induction-B}, $Q_D(w)$ is adjacent to $Q_D(u)$ but different from $Q_D(u^d)$,
 \item
  in Case~\ref{lemma:induction-C}, $w$ may be arbitrary,
  \item
  in Case~\ref{lemma:induction-D}, $Q_5$ is partitioned into subcubes $Q_D(w)$ and $Q_D(w^d)$.
\end{itemize}
Note that then
\begin{math}
|\bigcup M_D(w)|,|\bigcup M_D(w^d)|\le
\begin{cases}5& \text{in Cases \ref{lemma:induction-A}--\ref{lemma:induction-C},}\\
7 & \text{in Case \ref{lemma:induction-D}.}
\end{cases}
\end{math}
\newline
It follows that
\begin{align*}
 |\{x\in V(Q_D(w))\mid \{x,x^d\}\cap{\textstyle\bigcup}M=\emptyset\}|
 \ge&\ |V(Q_D(w))|-|{\textstyle\bigcup}M_D(w)|-|{\textstyle\bigcup}M_D(w^d)|\ge\\
 \ge& \ 2^4-14\ge2.
\end{align*}
Hence we can always select $x\in V(Q_D(w))$ (if $|M^2_d|$ is odd) or distinct $x,y\in V(Q_D(w))$ (if $M^2_d$ is empty)
such that $xx^d$ or $xx^d,yy^d$ are short edges of dimension $d$ not intersecting any edge of $M$. Hence Steps~\ref{construction:odd} and \ref{construction:empty} can be safely performed.
 
It remains to verify the validity of Steps~\ref{construction:left} and \ref{construction:right}. Note that both $M^0_d\cup P^0$ and $M^1_d\cup P^1$ are matchings in $K(Q^0)$ and $K(Q^1)$, respectively, inheriting the required partition into subcubes of dimension four, where each subcube --- in Cases~\ref{lemma:induction-A}, \ref{lemma:induction-B} and \ref{lemma:induction-C} --- contains no more than 7 vertices incident with edges of $M$. Moreover, the only subcubes with 6 or 7 such vertices might be 
\begin{itemize}
 \item
$Q_D(u)$ adjacent with $Q_D(w)$ in $Q^0$, and $Q_D(v)$ adjacent with $Q_D(w^d)$ in $Q^1$ (Case~\ref{lemma:induction-A}),
\item
$Q_D(u)$ adjacent with $Q_D(w)$ in $Q^0$, and $Q_D(w^d)$ in $Q^1$ (Case~\ref{lemma:induction-B}),
\item
$Q_D(w)$ in $Q^0$, and $Q_D(w^d)$ in $Q^1$ (Case~\ref{lemma:induction-C}),
\end{itemize}
without loss of generality assuming that $Q_D(w)$ lies in $Q^0=Q_{\overbar{\{d\}}}(0)$.
By the induction hypothesis, these matchings may be extended by short edges to cycles $C^i$ of lengths at least $\frac{3}{4}|V(Q^i)|$ for both $i\in\{0,1\}$. 

Case~\ref{lemma:induction-D} requires a special attention. Recall that in this case we had $Q_5$ partitioned into two subcubes such that $|Q_D(w)|,|Q_D(w^d)|\le7$. If $|M^2_d|$ was odd, Step~\ref{construction:odd} could have increased this number to 8. If $M^2_d$ was empty, both $|Q_D(w)|$ and $|Q_D(w^d)|$ were even and therefore at most 6. Step~\ref{construction:empty} would then increase this number to at most 8.  Consequently,  $M^0_d\cup P^0$ and $M^1_d\cup P^1$ in Steps~\ref{construction:left} and \ref{construction:right} are matchings in 4-dimensional subcubes $K(Q^0)$ and $K(Q^1)$, respectively, of sizes not exceeding $8/2=4$. By Lemma~\ref{lemma:n}, these matchings may be extended by short edges to cycles $C^i$ in $K(Q^i)$ of lengths at least $\frac{3}{4}|V(Q^i)|$ for both $i\in\{0,1\}$. 

In all cases, the desired cycle of length at least $\frac{3}{4}|V(Q_n)|$ extending $M$ by short edges exists by Observation~\ref{observation:on-construction}.
\end{proof}

It remains to use Theorem~\ref{thm:wiener} to transform the assumptions of the previous lemma into an upper bound on the matching size. 
\begin{theorem}
\label{thm:Induction}
Let $M$ be a matching in $K(Q_n)$, $n\ge2$, such that $|\bigcup M|< \frac{n^2}{6}+\frac{n}{2}+1$.  Then there is a set $S\subseteq E(Q_n)$ such that $M\cup S$  forms a cycle in $K(Q_n)$ of length at least $\frac{3}{4}|V(Q_n)|$. 
\end{theorem}
\begin{proof}
For $n\le8$ the statement of the theorem follows from Lemma~\ref{lemma:n}. For $n>8$ select a~set $\mathcal{S}$ such that $\bigcup M\subseteq\mathcal{S} \subseteq V(Q_n)$ and $|\mathcal{S}|=\lceil \frac{n^2}{6}+\frac{n}{2}\rceil$. 
Then 
$
\lceil\frac{n^2}{2|\mathcal{S}|-n-2}\rceil=4
$
and therefore by Theorem~\ref{thm:wiener}, there is a~set $D\subseteq[n]$ such that $|D|=4$ and $|\bigcup M_D(u)|\le|\mathcal{S}_D(u)|\le5$ for every $u\in\{0,1\}^{n-4}$. The statement of the theorem now follows from Lemma~\ref{lemma:induction}.
\end{proof}

Since the set of matchings in $K(Q_n)$ includes all matchings in $Q_n$, as a corollary we obtain the main result of this paper. 
\begin{corollary}
\label{cor:Main}
Every matching $M$ in $Q_n$ such that $n\ge2$ and  $|M|< \frac{n^2}{12}+\frac{n}{4}+\frac12$
can be extended to a~cycle of length at least $\frac{3}{4}|V(Q_n)|$. 
\end{corollary}
\acknowledgements
The author is much grateful to the anonymous referees for their helpful comments and suggestions that led to numerous improvements in the presentation of this paper. 

\end{document}